\documentclass[11pt]{article}
\usepackage{algorithm,algorithmic}

\newcommand{\ceil}[1]{\lceil #1 \rceil}
\newtheorem{lem}{Lemma}
\newtheorem{thm}{Theorem}

\newtheorem{prp}{Proposition}
\newtheorem{cor}{Corollary}
\newtheorem{dfn}{Definition}
\newcommand{\BBox}{\rule{0.1in}{0.1in}}
\newenvironment{proof}{\noindent {\bf Proof:}}{\ \BBox \\*}

\newcommand{\floor}[1]{\lfloor #1 \rfloor}

\begin{document}

\title{A survey of Chernoff and Hoeffding bounds}

\author{Alexandros V. Gerbessiotis\thanks{CS Department, New Jersey Institute of Technology, 
Newark, NJ 07102, USA. Email: alexg@njit.edu}
}
\maketitle
\thispagestyle{empty}


\begin{abstract}
This is a survey paper that discusses the original bounds 
of the seminal papers by Chernoff and Hoeffding.
Moreover, it includes a variety of derivative bounds in a
variety of forms. Complete proofs are provided as needed.
The intent is to provide  a repository of reference bounds
for the interested researcher.
\end{abstract}


\section{Chernoff's method}

In \cite{C52} bounds on the tails of a set of Bernoulli
trials are discussed in the form of Theorem 1.
Theorem~1 of \cite{C52} is restated below as
after some relevant definitions.
One can extract a variety of bounds out of Theorem~\ref{chthm1}. 
Such bounds can lead to  a sequence of lemmas such as
Lemma~\ref{lem1},
Lemma~\ref{lem2},
Lemma~\ref{lem3}, and
Lemma~\ref{lem4} for the right tails,
and
Lemma~\ref{lem5}, and
Lemma~\ref{lem6}, and
Lemma~\ref{lem7} for the left tails,
and their associated corollaries along with
some obvious concentration bounds 
(left and right tail bounds).

\begin{dfn}
Let $X_i$, $i=1, \ldots , n$ be  independent random 
variables identical in distribution to   
random variable $X$ whose moment generating function 
is $M_X (t) = E( e^{tX} )$,
and its cumulative distribution function is 
$F_X (x) = P( X \leq x)$.
Let  $E(X)=p=\mu$ and  define
$S_n = \sum_{i=1}^{n} X_i$ and $E(S_n) = n E(X)=n \mu$,
and
\[
m ( r ) = \inf E( e^{t(X-r) } ) 
   = \inf e^{-rt} M_X (t) 
   = \inf e^{-rt} E( e^{tX} )
.
\]
The infimum is with respect to the $t$'s values.
$M_X (t)$ attains a minimum value $m(0)$.
\end{dfn}
\noindent
Skipping some other details, the $t$ value for
the minimum is finite unless $P(X>0)=0$ or $P(X<0)=0$
and then $m(0)=P(X=0)$.
If $P( X \leq 0 ) > 0$ and
   $P( X \geq 0 ) > 0$ then $m(0) > 0$.
The following is Theorem 1 of Chernoff \cite{C52}

\begin{thm}[Chernoff \cite{C52}]
\label{chthm1}
If $E(X) < \infty$ and $r \geq E(X)$ then
\begin{equation}
\label{cher1}
 P( S_n \geq nr ) \leq (m(r))^n .
\end{equation}
If $E(X) > - \infty$ and $r \leq E(X)$ then
\begin{equation}
\label{chel1}
 P( S_n \leq nr ) \leq (m(r))^n .
\end{equation}
If $0 < u < m(r)$ and $E(X)$ might not exist,
\begin{equation}
\label{chel2}
  \lim_{n \rightarrow \infty} \frac{(m(r)- u)^n}{P(S_n \leq nr)} = 0.
\end{equation}
\end{thm}
\begin{proof}
In order to prove Eq.(\ref{cher1}) we  perform the following
transformation, using in the last step Markov's inequality,
and before that the monotonically increasing function
$X \mapsto e^{tX}$. Consider a $t>0$ and the monotonically
increasing function $f(x) = e^{tx}$.
\begin{eqnarray}
 P( S_n \geq nr ) &=&    P( t S_n    \geq t nr)  \nonumber \\
                 &=&    P( e^{tS_n} \geq e^{nrt} )\nonumber  \\
                 &\leq& \frac{ E( e^{tS_n} )}{e^{nrt}} \nonumber \\
\label{chel2a}
                 &\leq& e^{-nrt} E( e^{tS_n} ) \\
\label{chel2b}
                 &=   & e^{-nrt} M_{S_n} (t) .
\label{che1}
\end{eqnarray}
We then examine $ M_{S_n} (t) = E( e^{tS_n} )$. 
Since $S_n = \sum_i X_i$, then
\begin{eqnarray}
E( e^{tS_n} ) &=& E( e^{t \sum_{i} X_i } ) \nonumber \\
              &=& E( \prod_{i} e^{t  X_i } ) \nonumber \\
              &=&    \prod_{i} E( e^{t  X_i } ) \nonumber \\
              &=&    \prod_{i} M_{X_i} (t) \nonumber \\
              &=&     \left(   M_{X} (t) \right)^n .
\label{che2}
\end{eqnarray}
From Eq.~(\ref{che1}) by way of Eq.~(\ref{che2}) we obtain
the following.
\begin{eqnarray}
 P( S_n \geq nr ) 
    &\leq& e^{-nrt} E( e^{tS_n} )  \nonumber \\
    &\leq& e^{-nrt} \left(   M_{X} (t) \right)^n \nonumber \\
    & = & \left( e^{-rt}   M_{X} (t) \right)^n \nonumber \\
    & = & \inf_{t>0} \left( e^{-rt}   M_{X} (t) \right)^n \nonumber \\
    & = & \inf_{t>0} \left( m(r)  \right)^n .
\label{che3}
\end{eqnarray}
This completes the proof of Eq.(\ref{cher1}).
$ $ \\ $ $
The proof of Eq.(\ref{chel1}) is similar. Consider a $t < 0$ now.
\begin{eqnarray}
 P( S_n \leq nr ) &=&    P( t S_n    \geq t nr)  \nonumber \\
                 &=&    P( e^{tS_n} \geq e^{nrt} )\nonumber  \\
                 &\leq& \frac{ E( e^{tS_n} )}{e^{nrt}} \nonumber \\
                 &\leq& e^{-nrt} E( e^{tS_n} ) \nonumber \\
                 &=   & e^{-nrt} M_{S_n} (t) .
\label{che4}
\end{eqnarray}
The rest is identical to the previous case.
\end{proof}

In most of the discussion to follow we will assume
that $X_i$ and $X$ follow a Bernoulli distribution i.e.
$X_i \sim b(p)$ and $X \sim b(p)$ and thus $E(X_i ) = E(X)=p$,
where $0 < p < 1$. 
Then, $S_n \sim B(n,p)$.
Therefore we have the following.
\[
 M_X (t) = E( e^{tX} ) =  e^t \cdot p + 1 \cdot (1-p) 
         = 1 + p( e^t -1 ) \leq e^{pe^t -p}.
\]
The last part is because for all $x$, we have $e^x \geq 1+x$.

\begin{dfn}[Kullback-Leibler]
The Kullback-Leibler divergence $D_{KL}$ or just simply $D$
is defined as follows for two distributions of $n$ elements
$P= \cup_i \{ p_i \} $ , $p_i \geq 0$,
$Q= \cup_i \{ q_i \} $ , $q_i \geq 0$, 
$i=1, \ldots , n$ 
such that $\sum_i p_i = \sum_i q_i =1$.
\begin{equation}
\label{kldiv}
D_{KL} (P || Q ) = D( P || Q) = \sum_i p_i \ln{\frac{p_i}{q_i}}.
\end{equation}
\end{dfn}

\section{Derived right tails}

We derive the first Chernoff bound for a Binomial r.v.
$S_n$ which is the sum of $n$ Bernoulli random variables,
using Eq.(\ref{cher1}).

\begin{lem}
\label{lem1}
Let $X_i$ be independent random variables  following
the same distribution as random variable $X$, $X \sim b(p)$, 
where $0<p<1$.
Let $S_n = \sum_i X_i$.
Then, for any $r$ such that $p < r < 1$ we have the following.
\begin{equation}
\label{rche1}
 P( S_n \geq rn ) \leq \exp{( - D(r || p) n)}
                = \left[ \left( \frac{p}{r}     \right)^r
                         \left( \frac{1-p}{1-r} \right)^{1-r}
                  \right]^n .
\end{equation}
\end{lem}

Sometimes $r=p+t$ and the bound $p<r<1$ 
becomes $p<p+t <1$ or equivalently $0 < t < 1-p$. 
We report this variant in Corollary~\ref{cor2e}.

\begin{proof}(Of Eq.~(\ref{rche1}))
$ $ \\ $ $
By Eq.~(\ref{che3}) we have the following
 considering that $M_X (t) = E( e^{tX} ) 
         = 1 + p( e^t -1 ) \leq e^{pe^t -p}$.
\begin{eqnarray}
 P( S_n \geq nr ) 
   &\leq& \inf_{t>0} \left( e^{-rt}   M_{X} (t) \right)^n 
      \nonumber \\
   &\leq& \inf_{t>0} \left( e^{-rt}   (1 + p( e^t -1 )) \right)^n 
      \nonumber \\
\label{rche1a}
   &\leq& \inf_{t>0} \left( \frac{(1 + p( e^t -1 ))}{e^{rt}} \right)^n 
                \\
\label{rche1b}
   &\leq& \inf_{t>0} \left(  \exp{(f(t))} \right)^n 
\end{eqnarray}
As indicated by Eq.~(\ref{rche1b}),
$f(t) = \ln{( (1 + p( e^t -1 ))} -rt$.
Consider $f^{'} (t) = pe^t / (1+p e^t -p) -r$.
Equating to zero $f^{'} (t)=0$ and solving for $t$
we obtain $e^t = (1-p)r/(p(1-r))$.
Continuing with the second derivative we find
$f^{''} (t) = p(1-p) e^t / (1+pe^t -p)^2 >0$.
Therefore $f(t)$ has a minimum for $ e^t = (1-p)r/(p(1-r))$.
We then continue with Eq.~(\ref{rche1a}) as follows.
\begin{eqnarray}
 P( S_n \geq nr ) 
\label{rche1c}
   &\leq& \inf_{t>0} \left( \frac{(1 + p( e^t -1 ))}{e^{rt}} \right)^n 
\end{eqnarray}
The denominator $\exp{(rt)}$ for $e^t = (1-p)r/(p(1-r))$
is as follows.
\begin{equation}
\label{rche1d}
\exp{(rt)} = \left( (1-p)r/(p(1-r)) \right)^r =
           \frac{(1-p)^r r^r}{p^r (1-r)^r}.
\end{equation}
Likewise the numerator is as follows.
\begin{equation}
\label{rche1e}
(1 + p( e^t -1 )) = \frac{(1-p)rp}{p(1-r)}+1-p = \frac{1-p}{1-r}
\end{equation}
Therefore we have the following for the quantity below.
\begin{equation}
\label{rche1f}
\frac{(1 + p( e^t -1 ))}{e^{rt}} 
= \frac{1-p}{1-r} \cdot  \frac{p^r (1-r)^r}{(1-p)^r r^r}
                = \left[ \left( \frac{p}{r}     \right)^r
                         \left( \frac{1-p}{1-r} \right)^{1-r}
                  \right]^n .
\end{equation}
The proof is completed.
\end{proof}

We generate one more bound from Eq.(\ref{cher1}).
The bound is stronger than the Corollaries to follow that bring
it into an easier to deal with form.

\begin{lem}
\label{lem2}
Let $X_i$ be independent random variables  following
the same distribution as random variable $X$, $X \sim b(p)$, 
where $0<p < 1$.
Let $S_n = \sum_i X_i$.
Then, for any $r$ such that $p < r < 1$ we have the following,
after we substiture $r=(1+\delta ) p$, $\delta > 0$.
\begin{equation}
\label{rche2}
 P( S_n \geq rn ) =
 P( S_n \geq (1+ \delta) pn ) 
   \leq            \left(  \frac{e^r \cdot p^r}{e^p \cdot r^r} \right)^n 
     =             \left(  \frac{e^{\delta }}
                                  {(1+\delta)^{(1+\delta )}} 
                     \right)^{pn} 
\end{equation}
\end{lem}
\begin{proof}
By Eq.~(\ref{che3}) we use
$M_X (t) = E( e^{tX} )   = 1 + p( e^t -1 ) \leq e^{pe^t -p}$,
to obtain the following.
\begin{eqnarray}
 P( S_n \geq nr ) 
   &\leq& \inf_{t>0} \left( e^{-rt}   (1 + p( e^t -1 )) \right)^n 
      \nonumber \\
\label{rche2a}
   &\leq& \inf_{t>0} \left( e^{-rt}   e^{pe^t -p} \right)^n 
      \\
   &\leq& \inf_{t>0} \left( \exp{( pe^t -p-rt)} \right)^n 
      \nonumber \\
\label{rche2b}
   &\leq& \inf_{t>0} \left(  \exp{(f(t))} \right)^n 
\end{eqnarray}
The second to last part is because for all $x$, we have 
$e^x \geq 1+x$.
By Eq.~(\ref{rche2a}) and Eq.~(\ref{rche2b}) we
have $f(t) =  pe^t -p-rt$. Since $f^{'} (t) = pe^t -r$,
setting $f^{'} (t) = 0$ we obtain $t = \ln{(r/p)}$.
Moreover, $f^{''}(t) = p e^t$ is equal to
$f^{''} (  \ln{(r/p)} ) =r > 0$. Therefore $f(t)$
has a minimum at $t = \ln{(r/p)}$. Substituting this
value for $t$ in Eq.(\ref{rche2b}) the following
is obtained.

\begin{eqnarray}
 P( S_n \geq nr ) 
   &\leq& \inf_{t>0} \left( \exp{( pe^t -p-rt)} \right)^n 
      \nonumber \\
   &\leq&            \left(  \exp{( p (r/p) -p -r \ln{(r/p)}} \right)^n 
      \nonumber \\
   &\leq&            \left(  \frac{e^r \cdot p^r}{e^p \cdot r^r} \right)^n 
\label{rche2c}
\end{eqnarray}
Finally we substitute $r=(1 + \delta ) p$ in Eq.~(\ref{rche2c})
to obtain our result
\begin{eqnarray}
 P( S_n \geq nr ) 
   &\leq&            \left(  \frac{e^{(1+\delta )p} 
                              \cdot 
                                   p^{(1+\delta )p}}
                                  {e^p \cdot ((1+\delta)p)^{(1+\delta ) p}} 
                     \right)^n  \nonumber \\
\label{rche2d}
   &\leq&            \left(  \frac{e^{\delta }}
                                  {(1+\delta)^{(1+\delta )}} 
                     \right)^{pn} 
\end{eqnarray}
\end{proof}


A simpler proof of Eq.~(\ref{rche1}) for binomial random
variables is found in \cite{C79} and the proof is presented below.

\begin{lem}[\cite{C79}]
\label{lem3}
Let $X_i$ be independent Bernoulli random variables  following
the same distribution as random variable
 $X$, $X \sim b(p)$, where $0<p<1$.
Let $S_n = \sum_i X_i$.
Then, for any $r$ such that $p < r < 1$ we have the following.
\begin{equation}
\label{rche3}
 P( S_n \geq rn ) \leq \exp{( - D(r || p) n)}
                = \left[ \left( \frac{p}{r}     \right)^r
                         \left( \frac{1-p}{1-r} \right)^{1-r}
                  \right]^n .
\end{equation}
\end{lem}
\begin{proof}
Let $B(n,p,k) = \sum_{i=k}^n B(n,p;i) $.
For any $x \geq 1$, we have
\begin{eqnarray}
B(n,p,k) &\leq& \sum_{i=0}^{n} {n \choose i} p^{i} (1-p)^{n-i} x^{i-k} 
         \nonumber \\
         &\leq& x^{-k} \sum_{i=0}^{n} {n \choose i} p^{i} (1-p)^{n-i} x^{i} 
         \nonumber \\
         &\leq& x^{-k} \sum_{i=0}^{n} {n \choose i} (px)^{i} (1-p)^{n-i} 
         \nonumber \\
\label{rche3a}
         &\leq& x^{-k} (1+(x-1)p)^n.
\end{eqnarray}
For $r>p$ consider $x = (1-p)r /(p(1-r))$ and substitute for the $x$
of Eq.~(\ref{rche3a}). We obtain the following.
\begin{eqnarray}
B(n,p,k) 
    &\leq& x^{-k} (1+(x-1)p)^n  \nonumber \\
    &\leq& \left( \frac{(1-p)r}{p(1-r)}\right)^{-rn} \cdot
           \left( 1+ \left( \frac{(1-p)r}{p(1-r)}-1)p\right) \right)^{n} 
\label{rche3b}
\end{eqnarray}
The result then follows.
\end{proof}

We present  an alternative Chernoff bound formulation
which follows the simple case of Hoeffding bounds, where
$a_i = a =0$ and $b_i = b =1$. Naturally it applies
to Bernoulli or binary random variables (e.g. Rademacher).

\begin{lem}
\label{lem4}
Let $X_i$ be independent Bernoulli random variables  following
the same distribution as random variable $X$, $X \sim b(p)$,
where $0<p<1$.
Let $S_n = \sum_i X_i$.
Then, for any $r$ such that $p < r < 1$ we have the following.
\begin{equation}
\label{elem4}
 P( S_n \geq rn ) \leq 
                  \exp{\left( -2n (r-p)^2 \right) }.
\end{equation}
\end{lem}

\begin{proof}
The proof follows the steps of the proof of Theorem~\ref{chthm1}
to Eq.(\ref{chel2a}).
\begin{equation}
\label{elem4a}
 P( S_n \geq nr ) =    P( t S_n    \geq t nr) 
                  \leq \frac{ E( e^{tS_n} )}{e^{nrt}} 
                 \leq e^{-nrt} E( e^{tS_n} ) 
\end{equation}
We then proceed differently as follows.
\begin{eqnarray}
 P( S_n \geq nr ) 
   &\leq& e^{-nrt} E( e^{tS_n} )  \nonumber \\
   &\leq& e^{-nrt} e^{npt} E( e^{t( S_n -np)} )  \nonumber \\
   &\leq& e^{-nrt + npt} E(\prod_i e^{t( X_i - p)} )  \nonumber \\
\label{elem4b}
   &\leq& e^{-nrt + npt} \prod_i E( e^{t( X_i - p)} )  
\end{eqnarray}
We note that r.v. $X_i -p$ is bounded and $E(X_i -p)= E(X_i) -p =p-p=0$.
Then,  Proposition~\ref{phoef2} that is being utilized in Hoeffding
bounds can be used to show the following.
\begin{equation}
\label{elem4c}
   E( e^{t( X_i - p)} ) \leq \exp{\left( \frac{t^2}{8} \right)}.  
\end{equation}
Eq.(\ref{elem4b}) by way of Eq.(\ref{elem4c}) yields the following.
\begin{eqnarray}
 P( S_n \geq nr ) 
   &\leq& e^{-nrt + npt} \prod_i E( e^{t( X_i - p)} )  \nonumber\\
   &\leq& e^{-nrt + npt} \prod_i \exp{\left( \frac{t^2}{8} \right)} \nonumber\\
   &\leq& e^{-nrt + npt}  \exp{\left(n \frac{t^2}{8} \right)}  \nonumber\\
\label{elem4d}
   &\leq& e^{-nrt + npt + \frac{n t^2}{8} }. 
\end{eqnarray}
Consider the exponent of Eq.(\ref{elem4d}):
$f(t)= -nrt+ npt + nt^2 /8$. Its first derivative is
$f^{'} (t) = -nr+ np + t/4$. Equating it to zero and solving for
$t$ we have that
$f^{'} (t) = 0 = -nr+ np + t/4$ yields 
\begin{equation}
\label{elem4e}
t = 4(r-p) .
\end{equation}
Given that $f^{''} (t) = 1/4 > 0$, we have a minimum at
$t=4(r-p)$ for $f(t)$.
Therefore Eq.(\ref{elem4d}) by way of Eq.(\ref{elem4e})
yields the following.
\begin{eqnarray}
 P( S_n \geq nr ) 
   &\leq& e^{-nrt + npt + \frac{n t^2}{8} }  \nonumber \\
   &\leq& e^{-nr\cdot 4(r-p) + np\cdot 4(r-p) + 
           \frac{n \cdot 16(r-p)^2}{8} }  \nonumber \\
\label{elem4f}
   &\leq& e^{-2n(r-p)^2 }.
\end{eqnarray}
Eq.(\ref{elem4f}) is Eq.(\ref{elem4}) as needed.
\end{proof}

The following  corollary is more widely known
than Lemma~\ref{lem4}.

\begin{cor}
\label{cor4}
Let $X_i$ be independent Bernoulli  random variables  following
the same distribution as random variable $X$, $X \sim b(p)$,
where $0<p<1$.
Let $S_n = \sum_i X_i$.
Then, for any $r$ such that $p < r < 1$ we have the following.
\begin{equation}
\label{ecor4}
 P( S_n - E(S_n) \geq rn ) \leq 
                  \exp{\left( -2n r^2 \right) }.
\end{equation}
\end{cor}

\begin{proof}
In Lemma~\ref{lem4} substitute $r+p$ for $r$.
Then
$ P( S_n \geq (r+p) n ) = P( S_n \geq rn +pn) 
                        = P( S_n -E(S_n) \geq rn)$
and then substitute $r+p$ for $r$ in Eq.(\ref{elem4})
to obtain Eq.(\ref{ecor4}). 
\end{proof}

The following Corollary can be obtained from Lemma~\ref{lem2}.
It provides a more tangible upper bound than the generic
one of the Lemma.

\begin{cor}
\label{cor2a}
Let $X_i$ be independent random variables  following
the same distribution as random variable  $X$, $X \sim b(p)$,
where $0<p<1$.
Let $S_n = \sum_i X_i$.
Then, for any 
$\delta > 2e-1$.
we have the following,
\begin{equation}
\label{ecor2a}
 P( S_n \geq rn ) =
 P( S_n \geq (1+ \delta) pn ) 
   \leq 2^{-(1+ \delta ) pn}  .
\end{equation}
\end{cor}
\begin{proof}
From Lemma~\ref{lem2}, we have
\[
\frac{e^{\delta }}{(1+\delta)^{(1+\delta )}} 
\leq
\frac{e^{1+\delta }}{(2e)^{(1+\delta )}} 
\leq
\frac{e^{1+ \delta }}{(2e)^{1+\delta }} 
\leq 2^{-(1+\delta )}.
\]
The result follows then.
\end{proof}

The following Corollary is also obtained from Lemma~\ref{lem2}.
Note that $\delta > 0$ is better than the $\delta$ used in
Corollary~\ref{cor2c} that follows it.

The following inequality will be used \cite{AS}, \cite{ES}, \cite{NIST}.

\begin{lem}
\label{ineq1}
For every $ x  > -1$ and $x\neq 0$  we have the following.
\[
\frac{x}{1+x} < \ln{(1+x)}  < x.
\]
For every $x>0$ we have the following.
\[
 \ln{(x)} \leq x-1.
\]
For every $ x  < 1$ and $x\neq 0$  we have the following.
\[
x < - \ln{(1-x)} < \frac{x}{1-x} .
\]
\end{lem}

\begin{cor}
\label{cor2b}
Let $X_i$ be independent random variables  following
the same distribution as random variable $X$, $X \sim b(p)$,
where $0<p<1$.
Let $S_n = \sum_i X_i$.
Then, for any  $\delta > 0$ we have the following.
\begin{equation}
\label{ecor2b}
 P( S_n \geq (1+ \delta) pn ) 
 \leq \exp{\left( \frac{-\delta^2}{2+\delta} \cdot pn  \right)} .
\end{equation}
\end{cor}
\begin{proof}
We would like to upper bound the bound of Eq.(\ref{rche2}) 
of  Lemma~\ref{lem2}, as follows.
\[
\frac{e^{\delta }}{(1+\delta)^{(1+\delta )}} 
\leq \exp{ ( \frac{-\delta^2}{2+\delta} )}
\]
Consider function
$f(t)$ defined as follows.
\[
f(t) = \delta - (1+ \delta ) \ln{(1+ \delta)} + \frac{\delta^2}{2+\delta }.
\]
By Lemma~\ref{ineq1}
\[
\ln{(1+ \delta)} \geq \frac{2\delta}{2+\delta }.
\]
Therefore we have the following result
\[
f(t) \leq \delta - (1+ \delta ) \ln{(1+ \delta)} + 
\frac{\delta^2}{2+\delta }
\leq \frac{\delta (2 + \delta) -(1+\delta )(2 \delta) + \delta^2}
          {2+\delta } =0
\]
The Corollary then follows.
\end{proof}

A small improvement has been proposed by McDiarmid in the
following form.

\begin{cor}[\cite{McD2}]
\label{cor2bb}
Let $X_i$ be independent random variables  following
the same distribution as random variable $X$, $X \sim b(p)$,
where $0<p<1$.
Let $S_n = \sum_i X_i$.
Then, for any  $\delta > 0$ we have the following.
\begin{equation}
\label{ecor2bb}
 P( S_n \geq (1+ \delta) pn ) 
 \leq \exp{\left( \frac{-\delta^2}{2+2\cdot \delta / 3} \cdot pn  \right)} .
\end{equation}
\end{cor}
\begin{proof}
The proof utilizes the following inequality for all $x>0$.
\[
(1+x) \ln{(1+x)} -x \geq \frac{x^2}{2+(2/3)x}
\]
\end{proof}

The following
Corollary can also be obtained from Lemma~\ref{lem2}.
It is similar to the one in \cite{AV} for the binomial case.

The following inequality will be used.
\begin{lem}
\label{ineq7}
For every $ 0 \leq    \delta \leq 1$ we have the following.
\[
(1+ \delta ) \ln{(1+ \delta )} \geq   \delta + {\delta}^2 /3 .
\]
\end{lem}
\begin{proof}
For $ 0 \leq     x \leq 1$,
\begin{eqnarray*}
 \ln{(1+ x )}      &=& x -x^2 /2 + x^3 /3 - x^4 /4 + \ldots \\
(1+x) \ln{(1+ x )} &=& x -x^2 /2 + x^3 /3 - x^4 /4 + \ldots \\
                   &+& x^2 - x^3 /2 + x^4 /3 - x^5 /4 + \ldots \\
(1+x) \ln{(1+ x )} &\geq& x + x^2 /2 - x^3 /6
\end{eqnarray*}
We can then conclude since $0 \leq  x \leq 1$ that
\[
(1+x) \ln{(1+ x )} \geq x + x^2 /2 - x^3 /6  \geq x +x^2 /2 - x^2 /6
 \geq x + x^2 /3 .
\]
Set $x=\delta$ and the result follows.
\end{proof}

\begin{cor}
\label{cor2c}
Let $X_i$ be independent random variables  following
the same distribution as random variable  $X$, $X \sim b(p)$,
where $0<p<1$.
Let $S_n = \sum_i X_i$.
Then,  for any $\delta$ such that 
$0 < \delta < 1$ we have the following.
\begin{equation}
\label{ecor2c}
 P( S_n   \geq rn ) =
 P( S_n  \geq (1+ \delta ) pn ) 
 \leq  \exp{\left( \frac{-\delta^2}{3} \cdot pn  \right)} .
\end{equation}
\end{cor}
\begin{proof}
By Lemma~(\ref{ineq7})
we have for $0 < \delta < 1$,
$(1+\delta) \ln{(1+ \delta )}  \geq \delta + \delta^2 /3$.
This results to
\[
\frac{e^{\delta }}{(1+\delta)^{(1+\delta )}} 
\leq
\frac{e^{\delta }}{e^{\delta + \delta^2 /3}} 
\leq
e^{ - \frac{\delta^2 }{3} }.
\]
The result follows.
Note that as proved, $\delta <1$. However, we can 
improve the upper bound $\delta \leq 1$ by a more
tedious approach. We show it below.
Consider as before in Corollary~\ref{cor2b},
function $f(t)$ ($t$ substitutes for $\delta$) 
defined as follows.
\[
f(t) = t - (1+ t ) \ln{(1+ t)} + \frac{t^2}{3}.
\]
We would like to show $f(t) \geq 0$.
We first calculate its first derivative.
\[
f^{'} (t) = 2t/3 - \ln{(1+t)}.
\]
We note that $f^{'} (0) = 0$ and 
             $f^{'} (1) < 0$.
In order to study the monotonicity of
$f^{'} (t)$ we go on calculating the second derivative.
\[
f^{''} (t) = 2/3 - 1/(1+t).
\]
We note that $f^{''} (0) < 0$ for $t \leq 1/2$
and          $f^{''} (1) > 0$ for $t > 1/2$.
This means that $f^{'} (t)$ is monotonically
decreasing for $t \leq 1/2$ and given $f^{'} (0) = 0$,
negative for $t \leq 1/2$,
and monotonically increasing and since 
$f^{'} (1) < 0$ also negative for $1> t >    1/2$.
One can also separately confirm that $f(1/2) < 0 $.
Thus $f(\delta ) \leq 0$ for all $0< \delta \leq 1$.
The $\delta > 0$ is needed since $r > p$.
\end{proof}

We report below a corollary variant of Lemma~\ref{lem1}.
This is Corollary~\ref{cor2e}.

\begin{cor}
\label{cor2e}
Let $X_i$ be independent random variables  following
the same distribution as random variable $X$, $X \sim b(p)$,
where $0<p<1$.
Let $S_n = \sum_i X_i$.
Then, for any $t$ such that $0 < t < 1-p$ we have the following.
\begin{equation}
\label{rche1alt}
 P( S_n \geq (p+t)n ) \leq \exp{( - D((p+t) || p) n)}
                = \left[ \left( \frac{p}{p+t}     \right)^{p+t}
                         \left( \frac{1-p}{1-p-t} \right)^{1-p-t}
                  \right]^n .
\end{equation}
\end{cor}

\section{Derived left tails}

We proceed to deriving a Lemma identical to
Lemma~\ref{lem1} for the left tails. This is 
stated next.

\begin{lem}
\label{lem5}
Let $X_i$ be independent random variables  following
the same distribution as random variable $X$, $X \sim b(p)$,
where $0<p<1$.
Let $S_n = \sum_i X_i$.
Then, for any $r$ such that $0 < r < p$ we have the following.
\begin{equation}
\label{lche4}
 P( S_n \leq rn ) \leq \exp{( - D(r || p) n)}
                = \left[ \left( \frac{p}{r}     \right)^r
                         \left( \frac{1-p}{1-r} \right)^{1-r}
                  \right]^n .
\end{equation}
\end{lem}

Sometimes $r=p-t$, and the bound $0<r<p$ 
becomes $0<p-t <p$ or equivalently $0 < t < p$. 

\begin{proof}(Of Eq.~(\ref{lche4}))
$ $ \\ $ $
{\bf Method 1.}
A simple argument works as follows: the upper bound
on the number of successes generates a corresponding
lower bound on the number of failures. Thus
for $F_i = 1 - X_i $, we have $\sum F_i = n - \sum_i X_i$
or equivalently $Y_n = n - S_n$. Therefore
\[
P( S_n \leq rn ) 
       = P( n- S_n \geq n(1-r) ) 
       = P( Y_n \geq n(1-r) )
\]
The latter bound by Lemma~\ref{lem1} is 
bounded above by Eq.(\ref{rche1}) adjusting
it with $r$ replaced by the $1-r$
of $Y_n \geq n(1-r)$ and $p$ by $1-p$ to 
account for the failures not the successes of random
variable $Y_n$. The end result is that
\[
 P( S_n \leq rn ) \leq \exp{( - D((1-r) || (1-p)) n)} .
\]
However
$ D((1-r) || (1-p)) =   D(r || p) $ and therefore
\[
 P( S_n \leq rn ) \leq \exp{( - D((1-r) || (1-p)) n)} 
 = D(r || p) .
\]
\bigskip
$ $ \\ $ $
{\bf Method 2.}
Now let us reprove it following the method
of the proof of Lemma~\ref{lem1}.
Consider $t>0$, and apply the Chernoff trick and
finally use Markov's inequality as we did earlier.
\begin{eqnarray}
 P( S_n \leq nr ) &=&    P(  t S_n   \leq  t nr)  \nonumber \\
                  &=&    P( -t S_n   \geq -t nr)  \nonumber \\
                 &=&    P( e^{-tS_n} \geq e^{-nrt} )\nonumber  \\
                 &\leq& \frac{ E( e^{tS_n} )}{e^{nrt}} \nonumber \\
                 &\leq& e^{nrt} E( e^{-tS_n} ) \nonumber \\
                 &\leq& e^{nrt} (E( e^{-Xt} ))^n 
\label{che4a}
\end{eqnarray}
We calculate
 $ E( e^{-tX} ) 
         = 1 + p( e^{-t} -1 ) \leq e^{pe^{-t} -p}$.
The rest of the calculation are similarly to the ones before
\begin{eqnarray}
 P( S_n \leq nr ) 
   &\leq& \inf_{t>0} \left( e^{rt}   (1 + p( e^{-t} -1 )) \right)^n 
      \nonumber \\
\label{che4b}
   &\leq& \inf_{t>0} \left(  \exp{(f(t))} \right)^n  .
\end{eqnarray}
As indicated by Eq.~(\ref{che4b}),
$f(t) = \ln{(1 + p( e^{-t} -1 ))} +rt$.
Consider $f^{'} (t) = -pe^{-t} / (1+p^{-t} -p) +r$.
Equating to zero $f^{'} (t)=0$ and solving for $t$
we obtain $e^{-t} = (1-p)r/(p(1-r))$.
Continuing with the second derivative we find
$f^{''} (t) = p(1-p) e^{-t} / (1+pe^{-t} -p)^2 >0$.
Therefore $f(t)$ has a minimum for $ e^{-t} = (1-p)r/(p(1-r))$.
We then continue with Eq.~(\ref{che4b}) as follows. 
\begin{eqnarray}
 P( S_n \leq nr ) 
\label{che4c}
   &\leq& \inf_{t>0} \left( e^{rt} (1 + p( e^{-t} -1 )) \right)^n 
\end{eqnarray}
The term $\exp{(rt)}$ for $e^{-t} = (1-p)r/(p(1-r))$
is as follows.
\begin{equation}
\label{che4d}
\exp{(rt)} = \frac{p^r (1-r)^r}{(1-p)^r r^r}.
\end{equation}
Likewise the other term is as follows.
\begin{equation}
\label{che4e}
(1 + p( e^{-t} -1 )) = \frac{p(1-r)-p^2 (1-r) +p(1-p)r}{p(1-r)}
 = \frac{1-p}{1-r}
\end{equation}
Therefore we have the following for the quatinty below.
\begin{equation}
\label{che4f}
e^{rt} \cdot (1 + p( e^{-t} -1 ))
= \frac{1-p}{1-r} \cdot  \frac{p^r (1-r)^r}{(1-p)^r r^r}
                = \left[ \left( \frac{p}{r}     \right)^r
                         \left( \frac{1-p}{1-r} \right)^{1-r}
                  \right]^n .
\end{equation}
\end{proof}

We generate one more bound below.

\begin{lem}
\label{lem6}
Let $X_i$ be independent random variables  following
the same distribution as random variable $X$, $X \sim b(p)$,
where $0<p<1$.
Let $S_n = \sum_i X_i$.
Then, 
for any $r$ such that $0 < r < p$ or, equivalently,
for any $\delta$ such that  $0< \delta < 1$ we have
the following.
\begin{equation}
\label{lche5}
 P( S_n \leq rn ) =
 P( S_n \leq (1- \delta) pn ) 
   \leq            \left(  \frac{e^r \cdot p^r}{e^p \cdot r^r} \right)^n 
     =             \left(  \frac{e^{-\delta }}
                                  {(1-\delta)^{(1-\delta )}} 
                     \right)^{pn} 
\end{equation}
\end{lem}
\begin{proof}
We calculated earlier
 $ E( e^{-tX} ) 
         = 1 + p( e^{-t} -1 ) \leq e^{pe^{-t} -p}$.
By way of Eq.(\ref{che4b}) we have the following.
\begin{eqnarray}
 P( S_n \leq nr ) 
   &\leq& \inf_{t>0} \left( e^{rt}   (1 + p( e^{-t} -1 )) \right)^n 
      \nonumber \\
\label{che5a}
   &\leq& \inf_{t>0} \left( \exp{( pe^{-t} -p+rt)} \right)^n 
      \\
\label{che5b}
   &\leq& \inf_{t>0} \left(  \exp{(f(t))} \right)^n 
\end{eqnarray}
By Eq.~(\ref{che5a}) and Eq.~(\ref{che5b}) we
have $f(t) =  pe^{-t} -p+rt$. Since $f^{'} (t) = r-pe^{-t}$,
setting $f^{'} (t) = 0$ we obtain $t = \ln{(p/r)}$.
Moreover, $f^{''}(t) = p e^t$ is equal to
$f^{''} (  \ln{(p/r)} ) =r > 0$. Therefore $f(t)$
has a minimum at $t = \ln{(p/r)}$. Substituting this
value for $t$ in Eq.(\ref{che5a}) the following
is obtained.

\begin{eqnarray}
 P( S_n \leq nr ) 
   &\leq& \inf_{t>0} \left( \exp{( pe^{-t} -p+rt)} \right)^n 
      \nonumber \\
   &\leq&            \left(  \exp{( p r/p -p +r \ln{(p/r)}} \right)^n 
      \nonumber \\
   &\leq&            \left(  \frac{e^r \cdot p^r}{e^p \cdot r^r} \right)^n 
\label{che5c}
\end{eqnarray}
Finally we substitute $r=(1 - \delta ) p$ in Eq.~(\ref{che5c})
to obtain Eq.(\ref{lche5}).
\begin{eqnarray}
 P( S_n \leq nr ) 
   &\leq&            \left(  \frac{e^{(1-\delta )p} 
                              \cdot 
                                   p^{(1-\delta )p}}
                                  {e^p \cdot ((1-\delta)p)^{(1-\delta ) p}} 
                     \right)^n  \nonumber \\
\label{lche5d}
   &\leq&            \left(  \frac{e^{-\delta }}
                                  {(1-\delta)^{(1-\delta )}} 
                     \right)^{pn} 
\end{eqnarray}
\end{proof}

There is a weaker but more easier to deal bound for small
$p$. This
is shown next.
It is similar to the one in \cite{AV} for the binomial case.

The following inequality will be needed.

\begin{lem}
\label{ineq6}
For every $0 \leq \delta < 1$ we have the following.
\[
(1- \delta ) \ln{(1- \delta )} \geq  - \delta + {\delta}^2 /2.
\]
\end{lem}
\begin{proof}
For $0\leq x < 1$,
\begin{eqnarray*}
 \ln{(1- x )}      &=&  -  \sum_{i=1}^{\infty} x^i / i \\
                   &=& -x -x^2 /2 - x^3 /3 - x^4 /4 - \ldots \\
(1-x) \ln{(1- x )} &=& -x -x^2 /2 - x^3 /3 - x^4 /4 - \ldots \\
                   &+& +x^2 + x^3 /2 + x^4 /3 + x^5 /4 + \ldots \\
(1-x) \ln{(1- x )} &=& -x + x^2 /2 + x^2 /6 + \ldots
\end{eqnarray*}
The missing terms in the last form of the equation are positive.
We can then conclude
\[
(1-x) \ln{(1- x )} \geq -x + x^2 /2 + x^2 /6  \geq -x + x^2 /2
\]
\end{proof}

\begin{cor}
\label{cor5a}
Let $X_i$ be independent random variables  following
the same distribution as random variable $X$, $X \sim b(p)$,
where $0<p<1$.
Let $S_n = \sum_i X_i$.
Then, for any $\delta$ such that
 $0 < \delta < 1$ we have the following.
\begin{equation}
\label{ecor5a}
 P( S_n  \leq (1- \delta ) pn ) 
 \leq  \exp{\left( \frac{-\delta^2}{2} \cdot pn  \right)} .
\end{equation}
\end{cor}
\begin{proof}

By Lemma~(\ref{ineq6})
we have for every $0 < \delta < 1$,
$(1-\delta) \ln{(1- \delta )}  \geq -\delta - \delta^2 /2$.
This results to
\[
\frac{e^{-\delta }}{(1-\delta)^{(1-\delta )}} 
\leq
\frac{e^{-\delta }}{e^{-\delta - \delta^2 /2}} 
\leq
e^{ - \frac{\delta^2 }{2} }.
\]
The result follows.
\end{proof}


The symmetric case for the left tails
 to Lemma~\ref{lem4} is stated below.

\begin{lem}
\label{lem7}
Let $X_i$ be independent Bernoulli  random variables  following
the same distribution as random variable $X$, $X \sim b(p)$,
where $0<p<1$.
Let $S_n = \sum_i X_i$.
Then, for any $r$ such that $0 < r < p $ we have the following.
\begin{equation}
\label{elem7}
 P( S_n \leq rn ) \leq 
                  \exp{\left( -2n (r-p)^2 \right) }.
\end{equation}
\end{lem}
\begin{proof}
The proof is by symmetry to Lemma~\ref{lem4} for
$Y_i = 1 - X_i$ and $\sum_i Y_i = n-S_n$ instead.
\end{proof}

The following corollary  is then evident.
\begin{cor}
\label{cor7}
Let $X_i$ be independent Bernoulli  random variables  following
the same distribution as random variable $X$, $X \sim b(p)$,
where $0<p<1$.
Let $S_n = \sum_i X_i$.
Then, for any $r$ such that $0 < r < p$ we have the following.
\begin{equation}
\label{ecor7}
 P( S_n - E(S_n) \leq rn ) \leq 
                  \exp{\left( -2n r^2 \right) }.
\end{equation}
\end{cor}

\section{Derived  concentration bounds}

Finally the following
Corollary can also be obtained from 
Corollary~\ref{cor2c}
and 
Corollary~\ref{cor5a}.

\begin{cor}
\label{cor2d}
Let $X_i$ be independent random variables  following
the same distribution as random variable $X$, $X \sim b(p)$,
where $0<p<1$.
Let $S_n = \sum_i X_i$.
Then, for any $\delta$ such that $0 < \delta <1$ we have the following,
\begin{equation}
\label{ecor2d}
 P( |S_n -np|  \geq \delta pn ) 
 \leq 2 \cdot \exp{\left( \frac{-\delta^2}{3} \cdot pn  \right)} .
\end{equation}
\end{cor}
\begin{proof}
We have the following
\[
 P( |S_n -np|  \geq \delta pn ) 
 =
 P( S_n -np  \geq   \delta pn ) 
 +
 P( S_n -np  \leq - \delta pn ) 
 =
 P( S_n   \geq  (1+ \delta ) pn ) 
 +
 P( S_n   \leq  (1- \delta ) pn ) 
\]
By Corollary~(\ref{cor2c}) we bound 
$  P( S_n   \geq  (1+ \delta ) pn )$.
By Corollary~\ref{cor5a}  we bound 
$ P( S_n  \leq (1 - \delta ) pn ) $.
Simple manipulations show
$\exp{\left( \frac{-\delta^2}{2} \right)} <
  \exp{ ( \frac{-\delta^2}{3} )}$.
The result then follows.
\end{proof}


The following is a direct consequence of
    Corollary~\ref{cor4} and
    Corollary~\ref{cor7}.

\begin{cor}
\label{cor8}
Let $X_i$ be independent Bernoulli random variables   following
the same distribution as random variable $X$, $X \sim b(p)$,
where $0<p<1$.
Let $S_n = \sum_i X_i$.
Then, for any $r$ such that $0 < r < p $ we have the following.
\begin{equation}
\label{ecor8}
 P( |S_n -E(S_n) | \geq rn ) \leq  2 \cdot 
                  \exp{\left( -2n (r-p)^2 \right) }.
\end{equation}
\end{cor}

\section{Hoeffding's method}

We now present Theorem~1 and Theorem~2 of Hoeffding \cite{H63}.
The work \cite{H63} deals with random variables that are not
necessarily Bernoulli but are bounded e.g. $0 \leq X_i \leq 1$
or $a_i \leq X_i \leq b_i$. Theorem~1 \cite{H63}
for the case of Bernoulli trials takes the form of 
Corollary~\ref{cor2e} associated with a Chernoff bound.

That's why several times the terms Chernoff and Hoeffding 
refer to the same bound. Theorem~1 of \cite{H63} is the
strongest among the Hoeffding bounds.
The bounds are sufficient for large deviations.
Otherwise one  can use bounds available in \cite{B85}
and \cite{Feller1} \cite{Feller2}. Angluin-Valiant bounds \cite{AV} are
weaker but more useful; the latter are or may be more useful
for small $p$.

\section{Derived right tails}

Note that all random variables $X$ in the two theorems
that follow are to have finite first and second moments.
This Theorem~1 of \cite{H63} follows.
\begin{thm}
\label{hoef1}
Let $X_i$ be an independent random variable, $i=1, \ldots , n$.
Let $S_n = \sum_i X_i$, $\bar{X}=S_n /n$
and $p=E(\bar{X})$.
Then for any $h$ such that $0 < h < 1-p$ we have the following.
\begin{eqnarray}
 P( \bar{X}       -p  \geq h) )  
   &=& P( S_n \geq (p+h)n ) 
       \nonumber \\
&\leq& \exp{( - D((p+h) || p) n)}
       \nonumber \\
\label{ehoef1}
     &\leq&     \left[ \left( \frac{p}{p+h}     \right)^{p+h}
                       \left( \frac{1-p}{1-p-h} \right)^{1-p-h}
               \right]^n .
\end{eqnarray}
\end{thm}

Note that if $h > 1-p$, then Eq.~(\ref{ehoef1}) remains true,
and for $h=1-p$, the right-hand side can be replaced by
the limit $ h \rightarrow 1-p$ which is $p^n$.

We provide a proof below for the Theorem following the techniques
of the Chernoff's method also attributed to Cram{\'e}r \cite{C38}.
The theorem also appears as Theorem 5.1 in \cite{McD}.

\begin{proof}
In order to prove Eq.(\ref{ehoef1}) we  perform the following
transformation, using in the last step Markov's inequality,
and before that the monotonically increasing function
$X \mapsto e^{tX}$. Consider a $t>0$ and the monotonically
increasing function $f(x) = e^{tx}$. Note that $0 < h < 1-p$
below.
\begin{eqnarray}
 P( \bar{X}       -p  \geq h) )   &=&
 P( \frac{S_n}{n} - p \geq h )  \nonumber \\
   &=& P( S_n  - E(S_n )  \geq hn         )    \nonumber \\
   &=& P( S_n  -pn  \geq hn         )    \nonumber \\
    &=&    P( S_n   \geq (p+h)n ) 
\\
    &=&    P( e^{tS_n} \geq e^{t(p+h)} )\nonumber  \\
    &\leq& \frac{ E( e^{tS_n} )}{e^{n(p+h)t}} \nonumber \\
\label{hoe1a}
    & =  & e^{-n(p+h)t} E( e^{t S_n} ) 
\end{eqnarray}
We then examine $ E( e^{tS_n} )$. 
\begin{eqnarray}
 E( e^{tS_n} ) 
     &=& E( e^{t \sum_i X_i } ) \nonumber \\
\label{hoe1b}
     &=& \prod_{i=1}^{n} E( e^{t X_i} ) 
\end{eqnarray}
By the convexity of function $f(x) = \exp{(tx)}$ we have the
following: the line segment connecting the points
$(a, f(a))$ and $(b,f(b))$ lies over $f(x)$. The equation
of the line segment is as follows.
\begin{equation}
\label{hoe1c}
y- e^{ta} = \frac{e^{tb}-e^{ta}}{b-a} (x-a).
\end{equation}
Therefore,
\begin{equation}
\label{hoe1d}
y = e^{tb} \frac{x-a}{b-a} + e^{ta} \frac{b-x}{b-a} .
\end{equation}
Because of the convexity of $f$ we also have the following.
\begin{equation}
\label{hoe1e}
e^{tx} \leq y = e^{tb} \frac{x-a}{b-a} + 
                e^{ta} \frac{b-x}{b-a} .
\end{equation}
We now consider $E(e^{tX_i})$. We have the following.
\begin{equation}
\label{hoe1f}
E(e^{t X_i}) \leq  e^{tb} \frac{E(X_i)-a}{b-a} + 
                   e^{ta} \frac{b-E(X_i)}{b-a} .
\end{equation}
Because $0\leq X_i \leq 1$,  $a=0$ and $b=1$,
Eq.(\ref{hoe1f}) can be simplified.
\begin{equation}
\label{hoe1g}
E(e^{t X_i}) \leq  e^{t} E(X_i ) + e^{0} (1-E(X_i )) 
             \leq  e^t p_i + 1 - p_i ,
\end{equation}
where $p_i = E(X_i )$.
Then we plug Eq.(\ref{hoe1g}) into Eq.(\ref{hoe1b}).
We derived the following.
\begin{eqnarray}
 E( e^{tS_n} ) 
     &=& \prod_{i=1}^{n} E( e^{t X_i }) \nonumber \\
\label{hoe1h}
     &=& \prod_{i=1}^{n} (e^t p_i + 1 - p_i )  
\end{eqnarray}
Given that geometric means are at most their arithmetic means
we have the following.
\begin{eqnarray}
\label{hoe1i}
     \left( \prod_{i=1}^{n} (e^t p_i + 1 - p_i ) \right)^{1/n}
      \leq
                     \sum_i \frac{(e^t p_i + 1 - p_i )}{n}
      \leq
                     (e^t p + 1 - p ).
\end{eqnarray}
Therefore Eq.(\ref{hoe1i}) into Eq.(\ref{hoe1h}) yields
the following.
\begin{eqnarray}
 E( e^{tS_n} ) 
     &=& \prod_{i=1}^{n} (e^t p_i + 1 - p_i )  \nonumber \\
\label{hoe1j}
     &\leq& (e^t p + 1 - p )^n  \nonumber 
\end{eqnarray}

From Eq.(\ref{hoe1a}) utilizing Eq.(\ref{hoe1j}) we finally
derive the following.
\begin{eqnarray}
 P( \frac{S_n}{n} - p \geq h ) 
    &\leq& e^{-n(p+h)t} E( e^{t S_n} )  \nonumber \\
\label{hoe1k}
    &\leq& e^{-n(p+h)t} (e^t p + 1 - p )^n  \\
    &\leq& \exp{(-n(p+h)t + n \cdot \ln{(e^t p + 1 - p )})}  
           \nonumber \\
\label{hoe1l}
    &\leq& \exp{(g(t)) } 
\end{eqnarray}

We consider the monotonicity of function $g(t)$.
\begin{eqnarray}
\label{hoe1m}
g^{'} (t ) = - (p+h)n + n p e^t / \ln{(1-p+pe^t )}. 
\end{eqnarray}
Setting $g^{'} (t ) = 0$ and solving for $t$ we have,
\begin{equation}
\label{hoe1n}
t_0 = \ln{\frac{(p+h)(1-p)}{p(1-p-h)}},
\quad\quad
e^{t_0} = \frac{(p+h)(1-p)}{p(1-p-h)}.
\end{equation}
Moreover $ (p+h)(1-p) \geq p(1-p-h)$ and thus $t_0 > 0$,
since $h < 1-p$ and thus $1-p-h >0$.
Substituting Eq.(\ref{hoe1n}) for $t$ in $g(t)$ in
Eq.(\ref{hoe1l}) the result in the form of equation
Eq.(\ref{ehoef1}) follows.
\end{proof}

Theorem~1 of \cite{H63}
includes (weaker) upper bounds of Eq.(\ref{ehoef1})
of the form $\exp{(-nh^2 k(p))}$, where
\[
 k(p) = \frac{1}{1-2p} \ln{ \frac{1-p}{p} } ,
\]
for $ 0 < p < 1/2$, and
\[
 k(p) = \frac{1}{2p(1-p)} ,
\]
for $1/2 \leq p < 1$.
The proof arguments are tedious and we refer to \cite{H63}.
Furthermore, a weaker bound  of the form
$\exp{(-2n h^2 )}$ can also be derived. This can be
established also through Theorem~2 of \cite{H63}
that is stated and proved below.


Theorem~2 of \cite{H63}, utilizes the following
result that is proven separately.

\begin{prp}[Hoeffding \cite{H63}, Eq (4.16)]
\label{phoef2}
For a random variable $X$ such that $a \leq X \leq b$
with $E(X) = 0 $ and for any $t > 0$, we have the 
following.
\[
E [ e^{tX} ] \leq \exp{\left( \frac{t^2 (b-a)^2 }{8} \right) } .
\]
\end{prp}
\begin{proof}
The function $f(x) = \exp{(tx)}$ is a convex function.
Therefore  by Eq.(\ref{hoe1e}), we have the following.
\[
 e^{tX} \leq \frac{x-a}{b-a} e^{tb} +  \frac{b-x}{b-a} e^{ta} .
\]
Moreover, $E(X)=0$. By taking expectations on both sides we
conclude the following.
\begin{eqnarray*}
 E( e^{tX} ) &\leq& E( \frac{x-a}{b-a} e^{tb} +  \frac{b-x}{b-a} e^{ta} ) \\
             &\leq& \frac{E(x)-a}{b-a} e^{tb} +  \frac{b-E(x)}{b-a} e^{ta}\\
             &\leq& \frac{0   -a}{b-a} e^{tb} +  \frac{b-0   }{b-a} e^{ta}\\
             &\leq& \frac{-a}{b-a} e^{tb} +  \frac{b}{b-a} e^{ta}
\end{eqnarray*}
Consider
\begin{eqnarray*}
\label{etx}
E( e^{tX} ) &\leq& \exp{\left( \lg{\left(
                          \frac{-a}{b-a} e^{tb} +  \frac{b}{b-a} e^{ta}
                                 \right)
                                } \right)}
              =  \exp{(\ln{(g(t))})} .
\end{eqnarray*}
The function $g(t)$ will be rewritten as
$f(t) = \ln{(g(t))}$, and furthermore
by some change of variables,
using $p=b/(b-a)$ and therefore $1-p = -a/(b-a)$, and $x=(b-a)t$.
Then we have, after renaming variables with substitution, the following.
Note that $ta = (x/(b-a)) \cdot a =(x/(b-a)) \cdot (p-1)(b-a)=x(p-1)$.
\begin{eqnarray}
\label{fofx}
f(t) &=& \ln{(g(t))} \nonumber \\
f(t) &=& \ln{\left( \frac{-a}{b-a} e^{tb} +  \frac{b}{b-a} e^{ta} 
             \right)} \nonumber \\
     &=& \ln{\left( e^{ta} \left(
                \frac{-a}{b-a} e^{t(b-a)} +  \frac{b}{b-a} 
                \right)\right)} \nonumber \\
     &=& \ln{\left( e^{ta} \left(
                (1-p)  e^{t(b-a)} +  p
                \right)\right)} \nonumber \\
     &=& \ln{\left( e^{ta} \left(
                (1-p)  e^{x} +  p
                \right) \right)}\nonumber  \\
     &=& ta + \ln{\left( (1-p)  e^{x} +  p \right) } \nonumber \\
f(x) &=& x(p-1) + \ln{\left( (1-p)  e^{x} +  p \right) } 
\end{eqnarray}
Function $f(x)$ as indicated in Eq~\ref{fofx} has the
following properties:
$f(0)=0$ and $f^\prime (0)=0$.
We point out that
\[
 f^\prime (x) = p-1 + \frac{(1-p) e^x}{p+(1-p) e^x },
\]
and
\[
 f'' (x) =
          \frac{ (1-p) e^x (p+(1-p) e^x) -(1-p)e^x (1-p)e^x
                }{
                (p+(1-p) e^x )^2
                }
         =
          \frac{ p(1-p) e^x }{ (p+(1-p) e^x )^2 }.
\]
Using Taylor's formula we obtain
\[
 f(x) = f(0) + f^\prime (0) + f^{''} (r) x^2 / 2! ,
\]
for some $r$. We are going to find the maximum of 
$f'' (x)$. We observe that $f'' (x) = A \cdot B$, where
$A+B=1$, with $A=(p)/(p+(1-p) e^x )$ and 
              $B= ((1-p) e^x )/ (p+(1-p) e^x )$.
Therefore the second derivative is maximized for $A=B=1/2$
and $f'' (x) \leq 1/4$.
This implies.
\[
 f(x) = f(0) + f^\prime (0) + f^{''} (r) x^2 / 2!  \leq
  0 + 0 + (1/4) (1/2) x^2 \leq x^2 / 8.
\]
Moving backwards, Equation~(\ref{etx}) then yields, after
recovering the original variable names,
\begin{eqnarray*}
E( e^{tX} ) &\leq& \exp{(\ln{(g(t))})}  \\
            &\leq& \exp{
                        (  x^2 / 8 )
                       } \\
            &\leq& \exp{\left(
                         (b-a) t^2 / 8
                        \right)
                       } .
\end{eqnarray*}
\end{proof}

Theorem~2 of \cite{H63} is stated and proved below.
Whereas in Theorem~1 of \cite{H63} variables
$X_i$ were bounded in range by 0 and 1, the
bounds next are variable, in the sense  that
$a_i \leq X_i \leq b_i$.

\begin{thm}
\label{hoef2}
Let $X_i$ be an independent random variable, $i=1, \ldots , n$.
Let $S_n = \sum_i X_i$, $\bar{X}=S_n /n$, where
$a_i \leq X_i \leq b_i$
and $p=E(\bar{X})$ and $E(S_n) = np$.
Then for any $h$ such that $0< h < 1-p  $ we have the following.
\begin{eqnarray}
\label{ehoef2}
 P( \bar{X}       -p  \geq h) )  
   = P( S_n \geq (p+h)n ) 
   \leq 
          \exp{\left( 
           \frac{-2n^2 h^2 }{\sum_{i=1}^{n} (b_i - a_i )^2 }
           \right)}.
\end{eqnarray}
\end{thm}

\begin{proof}
In order to prove Eq.(\ref{ehoef2}) we  work out similarly
to the other bound of Eq.(\ref{ehoef1}).
Consider a $t>0$ and the monotonically
increasing and convex function $f(x) = e^{tx}$. 
\begin{eqnarray}
 P( \bar{X}       -p  \geq h) )  
    &=& P( S_n  - E(S_n )  \geq hn )  \nonumber \\
    &=&    P( e^{t( S_n -E(S_n))} \geq e^{thn)} )\nonumber  \\
    &\leq& \frac{ E( e^{t(S_n - E(S_n ))} )}{e^{htn}} \nonumber \\
\label{hoe2a}
    &\leq& e^{-htn} E( e^{t (S_n -E(S_n ))} )  \\
    &\leq&\inf_{h>0}  e^{-htn} E( e^{t (S_n -E(S_n ))} )  \nonumber
\end{eqnarray}
We then examine $ E( e^{t( S_n - E(S_n ))} )$. Note that $X_i$
are independent random variables for Eq.(\ref{hoe2b}) to be valid.
\begin{eqnarray}
 E( e^{t( S_n - E(S_n))} ) 
     &=& E( e^{t ( \sum_i X_i -E( \sum_i X_i )) } ) \nonumber \\
\label{hoe2b}
     &=& \prod_{i=1}^{n} E( e^{t (X_i - E(X_i ))} ) 
\end{eqnarray}
Since $E( X_i -E(X_i )) = 0$, the conditions of 
Proposition~\ref{phoef2} are satisfied.
Therefore 
\begin{equation}
\label{hoe2c}
E( e^{t (X_i -E(X_i ))} \leq 
\exp{\left( \frac{t^2 (b_i -a_i )^2 }{8} \right) } .
\end{equation}
Eq.(\ref{hoe2a}) by way of Eq.(\ref{hoe2b}) and
Eq.(\ref{hoe2c}) yields the following.
\begin{eqnarray}
 P( \bar{X}       -p  \geq h) )  
    &\leq&\inf_{h>0}  e^{-htn} 
                      E( e^{t (S_n -E(S_n ))} )  
      \nonumber \\
    &\leq&\inf_{h>0}  e^{-htn} \prod_{i=1}^{n} 
                      E( e^{t (X_i - E(X_i ))} ) 
      \nonumber \\
    &\leq&\inf_{h>0}  e^{-htn} \prod_{i=1}^{n} 
                      \exp{\left( \frac{t^2 (b_i -a_i )^2 }{8} \right) }
      \nonumber \\
\label{hoe2d}
    &\leq&\inf_{h>0}  \exp{\left(
                   -htn + t^2 \sum_{i=1}^{n} \frac{(b_i -a_i )^2 }{8}  
                           \right) } \\
\label{hoe2e}
    &\leq&\inf_{h>0}  \exp{\left( g(t)
                           \right) }
\end{eqnarray}

Function $t(t)$ is a parabola. Its minimum is for 
\begin{equation}
\label{hoe2f}
t_0 =  \frac{4hn}{ \sum_{i=1}^{n} (b_i -a_i )^2 }
\end{equation}
Substituting $t_0$ for $t$ in Eq.(\ref{hoe2d}) yields
the following equation.
\begin{eqnarray}
 P( \bar{X}       -p  \geq h)   
    &\leq&\inf_{h>0}  \exp{\left(
                   -htn + t^2 
                          \sum_{i=1}^{n} \frac{(b_i -a_i )^2 }{8}  
                           \right) } \\
      \nonumber \\
\label{hoe2g}
    &\leq&            \exp{\left( - \frac{2h^2 n^2}
                                    {\sum_{i=1}^{n} (b_i -a_i )^2 }
                           \right)}.
\end{eqnarray}
The proof  is complete as Eq.(\ref{hoe2g}) is Eq.(\ref{ehoef2}).
\end{proof}

\bigskip
By symmetry one can also prove the following e.g. by $Y_i =-X_i$.
\begin{thm}
\label{hoef3}
Let $X_i$ be an independent random variable, $i=1, \ldots , n$.
Let $S_n = \sum_i X_i$, $\bar{X}=S_n /n$, where
$a_i \leq X_i \leq b_i$
and $p=E(\bar{X})$ and $E(S_n) = np$.
Then for any $h$ such that $0< h  $ we have the following.
\begin{eqnarray}
\label{ehoef3}
 P( \bar{X}       -p  \leq -h) )  
   = P( S_n \leq (p-h)n ) 
   \leq 
          \exp{\left( 
           \frac{-2n^2 h^2 }{\sum_{i=1}^{n} (b_i - a_i )^2 }
           \right)}.
\end{eqnarray}
\end{thm}

In \cite{McD},
Theorem~\ref{hoef1}  appears as Theorem~5.1, and
Theorem~\ref{hoef2} and Theorem~\ref{hoef3}  
appear as Theorem~5.7 there.

There are variants of Theorem~\ref{hoef1} and
Theorem~\ref{hoef2}.
These include the following ones.

\begin{cor}
\label{choef2a}
Let $X_i$ be an independent random variable, $i=1, \ldots , n$,
such that $a \leq X_i \leq b$.
Let $S_n = \sum_i X_i$ and $E(S_n )=np$.
Then for any $\delta $ such that 
$0< \delta < (1-p)/p  $ we have the following.
\begin{eqnarray}
\label{echoef2a}
   P( S_n \geq (1+ \delta ) pn )
   &\leq& 
          \exp{\left( 
           \frac{-2n^2 \delta^2 p^2 }{n (b - a )^2 }
           \right)}.
\end{eqnarray}
\end{cor}
\begin{proof}
Set $h= \delta p$ in Theorem~\ref{hoef1}.
Moreover $a_i =a $, $b_i =b$, $i=1, \ldots , n$.
\end{proof}

A similar corollary to Corollary~\ref{choef2a} 
can be proven for 
$P( S_n \leq (1- \delta ) pn )$, if one bounds
the number of failures rather than successes with
Theorem~\ref{hoef2} or if $a_i =0$ and $b_i =1$,
equivalently consider $Y_i = - X_i$.
The upper bound would be identical then.

Corollary~\ref{choef4b}
appears in \cite{McD} as 5.3 
in Corollary 5.2 for $a=0$ and $b=1$.


\begin{cor}
\label{choef4b}
Let $X_i$ be an independent random variable, $i=1, \ldots , n$,
such that $a \leq X_i \leq b$.
Let $S_n = \sum_i X_i$ and $E(S_n )=np$.
Then for any $  h    $ such that 
$0<   h    <  n-np    $ we have the following.
\begin{eqnarray}
\label{echoef4b}
    P( S_n \geq pn +h  )
   &\leq& 
          \exp{\left( 
           \frac{-2h^2 }{\sum_{i=1}^{n} (b_i - a_i )^2 }
               \right)}.
\end{eqnarray}
\end{cor}
\begin{proof}
Replace $hn$ in Theorem~\ref{hoef2} with $h$.
\end{proof}

Moreover one can combine Corollary~\ref{choef4b}
                    and  Corollary~\ref{choef4c}
to bound $ P( | S_n - pn | \geq h )$.

%

The following is due to \cite{AV}.
It also appears as Corollary~\ref{cor2c} and
Corollary~\ref{cor5a}

\begin{thm}[Angluin-Valiant\cite{AV}]
\label{thmAV}
For every $n$, $p$, $b$ with $0 \leq p \leq 1$ and
$0 \leq b \leq 1$, we have the following.
\[
\sum_{k=0}^{k = \floor{(1-b)np}} B(n,p;k) \leq \exp{(-b^2 np /2 )},
\]
and
\[
\sum_{k = \ceil{(1+b)np}}^{n} B(n,p;k) \leq \exp{(-b^2 np /3 )}.
\]
\end{thm}

For the case of a binomial distribution of Bernouli trials
the following becomes available.
\begin{cor}
\label{choef12a}
Let $X_i$ be an independent random variable, $i=1, \ldots , n$,
such that $X_i \sim b(p)$, where $0<p <1$.
Let $S_n = \sum_i X_i$ and $E(S_n )=np$.
Then for any $\delta $ such that 
$0< \delta < (1-p)/p  $ we have the following.
\begin{eqnarray}
\label{echoef12a}
   P( S_n \geq (1+ \delta ) pn )
   &\leq& 
          \exp{\left( 
           -2n \delta^2 p^2 
           \right)}.
\end{eqnarray}
\end{cor}

\begin{cor}
\label{choef14b}
Let $X_i$ be an independent random variable, $i=1, \ldots , n$,
such that $ X_i \sim b(p)$, where $0<p <1$.
Let $S_n = \sum_i X_i$ and $E(S_n )=np$.
Then for any $  h    $ such that 
$0<   h    <  n-np    $ we have the following.
\begin{eqnarray}
\label{echoef14b}
    P( S_n \geq pn +h  )
   &\leq& 
          \exp{\left( 
           \frac{-2h^2 }{n}
               \right)}.
\end{eqnarray}
\end{cor}

\section{Derived left tails}

\begin{cor}
\label{choef2b}
Let $X_i$ be an independent random variable, $i=1, \ldots , n$,
such that $a \leq X_i \leq b$.
Let $S_n = \sum_i X_i$ and $E(S_n )=np$.
Then for any $\delta $ such that 
$0< \delta < 1 $ we have the following.
\begin{eqnarray}
\label{echoef2b}
   P( S_n \leq (1- \delta ) pn )
   &\leq& 
          \exp{\left( 
           \frac{-2n^2 \delta^2 p^2 }{n (b - a )^2 }
           \right)}.
\end{eqnarray}
\end{cor}
\begin{proof}
Bound the number of failures rather 
than successes with Theorem~\ref{hoef2} or 
equivalently consider $Y_i = - X_i$.
\end{proof}

Corollary~\ref{choef4c}
appears in \cite{McD} as 5.4 
in Corollary 5.2 for $a=0$ and $b=1$.

\begin{cor}
\label{choef4c}
Let $X_i$ be an independent random variable, $i=1, \ldots , n$,
such that $a \leq X_i \leq b$.
Let $S_n = \sum_i X_i$ and $E(S_n )=np$.
Then for any $  h    $ such that 
$0<   h    <  1-p     $ we have the following.
\begin{eqnarray}
\label{echoef4c}
    P( S_n \leq pn -h  )
   &\leq& 
          \exp{\left( 
           \frac{-2 h^2 }{\sum_{i=1}^{n} (b_i - a_i )^2 }
               \right)}.
\end{eqnarray}
\end{cor}

For the case of Bernoulli trials we have the following
simplifications.

\begin{cor}
\label{choef22b}
Let $X_i$ be an independent random variable, $i=1, \ldots , n$,
such that $ X_i \sim b(p)$, where $0<p <1$.
Let $S_n = \sum_i X_i$ and $E(S_n )=np$.
Then for any $\delta $ such that 
$0< \delta < 1 $ we have the following.
\begin{eqnarray}
\label{echoef22b}
   P( S_n \leq (1- \delta ) pn )
   &\leq& 
          \exp{\left( 
           -2n \delta^2 p^2 
           \right)}.
\end{eqnarray}
\end{cor}

\begin{cor}
\label{choef24c}
Let $X_i$ be an independent random variable, $i=1, \ldots , n$,
such that $\leq X_i \sim b(p)$, where $0<p<1$.
Let $S_n = \sum_i X_i$ and $E(S_n )=np$.
Then for any $  h    $ such that 
$0<   h    <  1-p     $ we have the following.
\begin{eqnarray}
\label{echoef24c}
    P( S_n \leq pn -h  )
   &\leq& 
          \exp{\left( 
           \frac{-2 h^2 }{n}
               \right)}.
\end{eqnarray}
\end{cor}

\section{Derived concentration bounds}

The following Corollary combines Corollary~\ref{choef2a}
with Corollary~\ref{choef2b}.

\begin{cor}
\label{choef2c}
Let $X_i$ be an independent random variable, $i=1, \ldots , n$,
such that $a \leq X_i \leq b$.
Let $S_n = \sum_i X_i$ and $E(S_n )=np$.
Then for any $\delta $ such that 
$0< \delta < 1 $ we have the following.
\begin{eqnarray}
\label{echoef2c}
   P( |S_n -np| \geq \delta  pn )
   &\leq& 
       2  \exp{\left( 
           \frac{-2n^2 \delta^2 p^2 }{n (b - a )^2 }
           \right)}.
\end{eqnarray}
\end{cor}

The following Corollary combines Corollary~\ref{choef4b}
with Corollary~\ref{choef4c} instead.

\begin{cor}
\label{choef4d}
Let $X_i$ be an independent random variable, $i=1, \ldots , n$,
such that $a \leq X_i \leq b$.
Let $S_n = \sum_i X_i$ and $E(S_n )=np$.
Then for any $  h    $ such that 
$0<   h    <  1-p     $ we have the following.
\begin{eqnarray}
\label{echoef4d}
    P( | S_n -pn| \geq h  )
   &\leq& 
       2  \exp{\left( 
           \frac{-2 h^2 }{\sum_{i=1}^{n} (b_i - a_i )^2 }
               \right)}.
\end{eqnarray}
\end{cor}




\newpage
\begin{thebibliography}{0}
%
\bibliographystyle{amsalpha}


\bibitem{AS}
     M.~Abramowitz and I.~A.~Stegun.
     Handbook of mathematical functions 
with formulas, graphs, and mathematical tables. 
New York: Dover Publications.  Ninth printing.

\bibitem{AV}
     D.~Angluin and L.~G.~Valiant.
     Fast probabilistic algorithms for 
hamiltonian circuits and matchings. 
Journal of Computer and System Sciences.
Volume 18, Issue 2, 1979, Pages 155-193.

\bibitem{B85}
     B.~B{\'e}la Bollob{\'a}s.
     Random Graphs.
Academic Press 1985.



\bibitem{C52}
     H.~Chernoff. 
     A measure of asymptotic efficiency for 
tests of a hypothesis based on the 
sum of observations. 
The Annals of Mathematical Statistics. 23 (4):493-507, 1952.

\bibitem{C79}
     V.Chv{\'a}tal.
     The tail of the hypergeometric distribution. 
Discrete Mathematics. 25(3):285-287, 1979, Elsevier.



\bibitem{C38}
     H.~Cram{\'e}r.
     Sus un nouveau th{\`e}or{\'e}me-limite de la 
th{\'e}orie des probabilit{\'e}s.
Actualit{\'e}s Scientifiques et Industrielles.
No 736. Paris 1938.

\bibitem{ES}
     P.~Erdos and J.~Spencer.
     Probabilistic methods in combinatorics. 
Academic Press. New York, 1974.

\bibitem{Feller1}
    W.~Feller.
     An introduction to probability theory 
and its applications. 
Vol. 1. 3rd Edition. John Wiley \& Sons. 
New York, 1968. 

\bibitem{Feller2}
    W.~Feller.
     An introduction to probability theory 
and its applications. 
Vol. 2. John Wiley \& Sons. New York, 1971. 



\bibitem{H63}
     W.~Hoeffding.
     Probability inequalities for sums 
of bounded random variables. 
Journal of the American Statistical Association. 
58 (301):13-30, March 1963. 

\bibitem{KL51}
     S.~Kullback and R.~A.~Leibler.
     On information and sufficiency.
The Annals of Mathematical Statistics. 22 (1):79-86, 1952.


\bibitem{McD}
    C.~McDiarmid . 
On the method of bounded differences. 
In: Siemons J, ed. Surveys in Combinatorics, 
1989: Invited Papers at the Twelfth British Combinatorial 
Conference. London Mathematical Society 
Lecture Note Series. Cambridge University Press; 
1989:148-188. 

\bibitem{McD2}
  C.~McDiarmid . 
Concentration.
In: Probabilistic methods for algorithmic
discrete mathematics.
M.~Habib, C.~McDiarmid, J.~Ramirez-Alfonsin, B.~Reed (Eds),
Algorithms and Combinatorics 16, Springer, 1998.

\bibitem{NIST}
National Institute of Standards and Technology. https://dlmf.nist.gov/4.5

\end{thebibliography}
\end{document}